\documentclass[runningheads, envcountsame, a4paper]{llncs}
\usepackage{subcaption}
\usepackage{color}
\usepackage{amsmath}
\usepackage[utf8]{inputenc}
\usepackage{amssymb}
\usepackage[ruled, vlined]{algorithm2e} 
\usepackage[noend]{algpseudocode} 
\usepackage{graphicx}

\newsavebox{\figurebox}
\newlength{\figureboxheight}

\mainmatter              

\title{Greedy Shortest Common Superstring Approximation in Compact Space}
\titlerunning{Greedy Shortest Common Superstring Approximation in Compact Space} 

\toctitle{Greedy Shortest Common Superstring Approximation in Compact Space}
\tocauthor{Jarno~Alanko and Tuukka~Norri}

\author{Jarno Alanko \and Tuukka Norri}
\authorrunning{J. Alanko and T. Norri} 

\institute{Department of Computer Science \\ University of Helsinki \\
Gustaf Hällströmin katu 2b, 00560 Helsinki, Finland \\ 
\email{jarno.alanko@helsinki.fi}, \email{tuukka.norri@helsinki.fi}}

\begin{document}

\maketitle

\setcounter{footnote}{0}

\begin{abstract}
Given a set of strings, the shortest common
superstring problem is to find the shortest
possible string that contains all the input
strings. The problem is NP-hard, but a 
lot of work has gone into designing approximation
algorithms for solving the problem. 
We present the first time and space efficient 
implementation of the classic greedy heuristic which merges
strings in decreasing order of overlap 
length. Our implementation works in $O(n \log \sigma)$ time and bits of space,
where $n$ is the total length of the input
strings in characters, and $\sigma$ is the size of the
alphabet. After index construction, a practical implementation of our algorithm
uses roughly $5 n \log \sigma$ bits of space
and reasonable time for a real dataset that consists
of DNA fragments.
\end{abstract}

\section{Introduction}

Given a set of strings, the shortest common superstring is the shortest
string which contains each of the input strings as a substsring.
The problem is NP-hard \cite{gallant1980finding}, but efficient approximation algorithms exist.
Perhaps the most practical of the approximation algorithms is the greedy algorithm first analyzed by
 Tarhio, Ukkonen \cite{tarhio1988greedy} and Turner \cite{turner1989approximation}. The algorithm greedily joins together the pairs of strings
with the longest prefix-suffix overlap, until only one string remains. In case there are equally long overlaps,
the algorithm can make an arbitrary selection among those. The
remaining string is an approximation of the shortest common superstring. The algorithm has been proven to give a superstring
with length at most $3 \frac{1}{2}$ times the optimal length
\cite{kaplan2005greedy}. It was originally conjectured by
Ukkonen and Tarhio \cite{tarhio1988greedy}
that the greedy algorithm never outputs a superstring that is more 
than twice as long as the optimal, and the conjecture is still
open.

Let $m$ be the number of strings, $n$ be the sum of the lengths of all the strings, and $\sigma$ the size of the alphabet.
In 1990 Ukkonen showed how to implement the greedy algorithm in $O(n)$ time
and $O(n \log n)$ bits of space using the Aho-Corasick automaton
\cite{ukkonen1990linear}. Since
then, research on the problem has focused on finding algorithms with better provable approximation ratios (see e.g.
\cite{mucha2013lyndon} for a summary). 
Currently, algorithm with the best proven approximation ratio in peer reviewed literature 
is the one by Mucha with an approximation
ratio of $2 \frac{11}{23}$ \cite{mucha2013lyndon}, and there is a preprint claiming
an algorithm with a ratio of $2 \frac{11}{30}$ \cite{paluch2014better}.
However, we are not aware of any published algorithm that solves
the problem in better than $O(n \log n)$ bits of space.
Improving
the factor $\log n$ to $\log \sigma$ is important
in practice. Many of the largest data
sets available come from DNA strings which have
an alphabet of size only 4, while $n$ can be over $10^9$.

We present an algorithm that implements the greedy 
heuristic in $O(n \log \sigma)$ time and bits of space. 
It is based on the FM-index enhanced
with a succinct representation of the topology of the suffix
tree. The core of the algorithm is the iteration of
prefix-suffix overlaps of input strings in decreasing order
of length using a technique described in \cite{makinen2015genome} and \cite{simpson2010efficient}, 
combined with Ukkonen's bookkeeping \cite{ukkonen1990linear} to keep track of the paths formed in the
overlap graph of the input strings. The main technical novelty of this work is the implementation of Ukkonen's
bookkeeping in $O(n \log \sigma)$ space. We also have a working implementation of the algorithm
based on the SDSL-library \cite{gog2014theory}. For practical reasons the implementation differs
slightly from the algorithm presented in this paper, but the time and space usage should be
similar.

\section{Preliminaries}

Let there be $m$ strings $s_1, \ldots, s_m$ drawn from the alphabet $\Sigma$ of
size $\sigma$ such that
the sum of the lengths of the strings is $\sum_{i=1}^m |s_i| = n$. We build a single string
by concatenating the $m$ strings, placing a separator character 
\$ $\not\in \Sigma$ between each string. 
We define that the separator is lexicographically smaller than all characters
in $\Sigma$.
This gives us the string
$S = s_1 \$ s_2 \$ \cdots s_m \$$ of length $n + m$.
Observe that the set of suffixes that are prefixed by some substring $\alpha$
of $S$
are adjacent in the lexicographic ordering of the suffixes. We call this
interval in the sorted list of suffixes the \emph{lexicographic range} of string $\alpha$.
All occurrences of a substring $\alpha$ can be uniquely represented as a triple
$(a_\alpha, b_\alpha, d_\alpha)$, where $[a_\alpha, b_\alpha]$ is the
lexicographic range of $\alpha$, and $d_\alpha$ is the length of $\alpha$.
A string $\alpha$ is \emph{right maximal} in $S$
if and only if there exist two or more distinct characters 
$y,z \in \Sigma \cup \{\$\}$ such that
the strings $\alpha y$ and $\alpha z$ are substrings of $S$.
Our algorithm needs support for two operations on substrings:
left extensions and suffix links.
A left extension of string $\alpha$ with character 
$x$ is the map
$(a_\alpha, b_\alpha, d_\alpha) \mapsto (a_{x\alpha}, b_{x\alpha},
d_{x\alpha})$. A suffix link for the right-maximal string $x \alpha$
is the map $(a_{x\alpha}, b_{x\alpha}, d_{x\alpha})
\mapsto (a_\alpha, b_\alpha, d_\alpha)$.


\section{Overview of the Algorithm}
We use Ukkonen's 1990 algorithm \cite{ukkonen1990linear} as a basis
for our algorithm. Conceptually, we have a complete directed
graph where vertices
are the input strings, and the weight of the edge from string
$s_i$ to string $s_j$ is the length of the longest suffix of $s_i$ which
is also a prefix of $s_j$. If there is no such overlap, the weight
of the edge is zero. The algorithm finds a 
Hamiltonian path over the graph, and merges
the strings in the order given by the path to form the 
superstring. We define the merge
of strings $s_i = \alpha \beta$ and $s_j = \beta \gamma$, where 
$\beta$ is the longest prefix-suffix overlap of $s_i$ and $s_j$, as
the string $\alpha\beta\gamma$. 
It is known that the string formed by merging the strings in
the order given by the maximum weight Hamiltonian path gives a
superstring of optimal length \cite{tarhio1988greedy}. 
The greedy algorithm tries to heuristically find a 
Hamiltonian path with a large total length. 

Starting from a graph $G$
where the vertices are the input strings and there are no edges,
the algorithm iterates all prefix-suffix overlaps of pairs
of strings in decreasing order of length. For each pair $(s_i,s_j)$
we add an edge from $s_i$ to $s_j$ iff the in-degree of $s_j$ is
zero, the out-degree of $s_i$ is zero, and 
adding the edge would not create a cycle in $G$.
We also consider overlaps of length zero, so every possible edge is considered and
it is easy to see that in the end the added edges form a Hamiltonian 
path over $G$.
\section{Algorithm} \label{sec:algorithm}

Observe that if an input string is a proper substring of another input string, then any valid superstring that contains the longer string
also contains the shorter string, so we can always discard the
shorter string. Similarly if there are strings that occur multiple
times, it suffices to keep only one copy of each. This preprocessing
can be easily done in $O(n \log \sigma)$ time and space 
for example by backward searching all the input strings
using the FM-index.

After the preprocessing, we sort the input strings into lexicographic
order, concatenate them
placing dollar symbols in between the strings, and build
an index that supports suffix links and left extensions. The sorting
can be done with merge sort such that string comparisons are done
$O(\log(n))$ bits at a time using machine word level parallelism, as allowed 
by the RAM model. This works in $O(n \log \sigma)$ time and space if the sorting is
implemented so that it does not move the strings around, but instead manipulates only
pointers to the strings.

For notational convenience, from here on
$s_i$ refers to the string with lexicographic rank $i$ among the input strings.

We iterate in decreasing order of length 
all the suffixes of the input strings $s_i$ that occur at least twice in $S$
and for each check whether the suffix is also a prefix of some other string $s_j$,
and if so, we add an edge from $s_i$ to $s_j$ if possible.
To enumerate the prefix-suffix overlaps, we use the key ideas from the algorithm for reporting
all prefix-suffix overlaps to build an overlap graph described in \cite{makinen2015genome} and \cite{simpson2010efficient}, 
adapted to get the overlaps in decreasing order of length.

We maintain an iterator for each of the input strings.
An iterator for the string $s_i$ is a quadruple $(i,\ell,r,d)$, where $[\ell,r]$ 
is the lexicographic range of the 
current suffix $\alpha$ of $s_i$ and $d$ is the length of $\alpha$,
i.e. the depth of the iterator.
Suffixes of the input strings which are
not right maximal in the concatenation $S = s_1\$\ldots s_m\$$ can never be a prefix
of any of the input strings.
The reason is that if $\alpha$ is not right-maximal,
then $\alpha$ is always followed by 
the separator $\$$. This means that if $\alpha$
is also a prefix of some other string $s_j$, then $s_j = \alpha$, because
the only prefix of $s_j$ that is followed by a $\$$ is the whole string $s_j$. 
But then $s_j$ is a substring of $s_i$, which can not happen because 
all such strings were removed in the preprocessing stage.
Thus, we can safely disregard any suffix $\alpha$ of $s_i$ 
that is not right maximal
in $S$. Furthermore, if a suffix $\alpha$ of $s_i$ is not right maximal, then none
of the suffixes $\beta\alpha$ are right-maximal either, 
so we can disregard those, too.

We initialize the iterator for each string $s_i$ by backward searching
$s_i$ using the FM-index for as long as the current suffix of $s_i$ is right-maximal.
Next we sort these quadruples in the decreasing order of 
depth into an array $\texttt{iterators}$.
When this is done, we start iterating from the iterator
with the largest depth, i.e. the first element of $\texttt{iterators}$.
Suppose the current iterator
corresponds to string $i$, and the current suffix of string $s_i$ is $\alpha$. 
At each step of the iteration we check whether $\alpha$ is also a prefix of some 
string by executing a left extension with the separator character
$\$$. If the lexicographic range $[\ell',r']$ of $\$\alpha$ is non-empty,
we know that the suffixes of $S$ in the range $[\ell',r']$ start with a dollar
and are followed by a string that has $\alpha$ as a prefix. We conclude that
the input string with lexicographic rank $i$ among the input strings
has a suffix of length 
$d$ that matches a prefix of the strings with lexicographic ranks
$\ell' , \ldots , r'$ among the input strings. This is true because
the lexicographic order of the suffixes of $S$ that start with dollars
coincides with the lexicographic ranks of the strings following the
dollars in the concatenation, because the strings are concatenated
in lexicographic order.

Thus, according to the greedy heuristic, 
we should try to merge $s_i$ with a string from
the set $s_{\ell'},\ldots,s_{r'}$, which corresponds to
adding an edge from $s_i$ to some string from 
$s_{\ell'},\ldots,s_{r'}$ in the graph $G$. We
describe how we maintain the graph $G$ in a moment. After updating
the graph, we  update the current iterator by decreasing $d$ by 
one and taking a suffix link of the lexicographic range $[\ell, r]$. 
The iterator with the next largest $d$ can be found in constant time
because the array $\texttt{iterators}$ is initially sorted in
descending order of depth. We can maintain
a pointer to the iterator with the largest $d$. If at some step
$\texttt{iterators}[k]$ has the largest depth, then in the next
step either $\texttt{iterators}[k+1]$ or
$\texttt{iterators}[1]$ has the largest depth.
The pseudocode for the main iteration loop is shown in Algorithm
\ref{alg:main_iteration}.

\begin{algorithm}

\SetAlgoLined
\DontPrintSemicolon
$k \gets 1$ \;
\While{\upshape $\texttt{iterators}[k].d \geq 0$}{

	$(i,[\ell,r],d) \gets \texttt{iterators}[k]$ \;
	$[\ell', r'] \gets \texttt{leftextend}([\ell,r], \$)$ \;
	\If{$[l', r']$ is non empty}{
		$\texttt{trymerge}([l', r'], i)$
	}
	$\texttt{iterators}[k] \gets (i, \texttt{suffixlink}(\ell, r), d-1)$ \;
	\eIf{$i = m$ or \upshape $(\texttt{iterators}[1].d > \texttt{iterators}[i+1].d)$}{
		$k \gets 1$
	}{
	\textbf{else} $k \gets k +1 $ 
	}
}
\caption{Iterating all prefix-suffix overlaps
\label{alg:main_iteration}}
\end{algorithm}

%

{\noindent}Now we describe how we maintain the graph $G$.
The range $[\ell', r']$ now represents the
lexicographical
ranks of the input strings that are prefixed by $\alpha$
among all input strings. Each string $s_j$ in this
range is a candidate to merge to string $s_i$, but some
bookkeeping is needed to keep track of available strings.
We use
essentially the same method as Tarhio and
Ukkonen \cite{tarhio1988greedy}.
We have bit vectors $\texttt{leftavailable}[1..m]$ and
$\texttt{rightavailable}[1..m]$ such that
$\texttt{leftavailable}[k] = 1$ if and only if string 
$s_k$ is available
to use as the left side of a merge, and
$\texttt{rightavailable}[k] = 1$ if and only if string $s_k$
is available as the right side of a merge. Equivalently,
$\texttt{leftavailable}[k] = 1$ iff the out-degree of $s_k$ is
zero and $\texttt{rightavailable}[k] = 1$ if the in-degree of
$s_k$ is zero.
Also, to prevent the formation
of a cycle, we need arrays
$\texttt{leftend}[1..m]$, where $\texttt{leftend}[k]$ 
gives the leftmost string
of the chain of merged strings to the left of $s_k$, and
$\texttt{rightend}[1..m]$, where $\texttt{rightend}[k]$ gives the
rightmost string of the chain of merged strings to the right of $s_k$.
We initialize $\texttt{leftavailable}[k] =
\texttt{rightavailable}[k] = 1$ and 
$\texttt{leftend}[k] = \texttt{rightend}[k] = k$ for all $k = 1, \ldots, m$.

When we get the interval $[\ell', r']$ such that
$\texttt{leftavailable}[j] = 1$,
we try to find an index
$j \in [\ell_{\$ \alpha}, r_{\$ \alpha}]$ such that $\texttt{rightavailable}[i] = 1$ and 
$\texttt{leftend}[j] \neq i$. Luckily we only need to examine at most
two indices $j$ and $j'$ such that 
$\texttt{rightavailable}[j] = 1$ and $\texttt{rightavailable}[j'] = 1$ because
if $\texttt{leftend}[j] = i$, then $\texttt{leftend}[j'] \neq i$, and vice versa. This procedure
is named $\texttt{trymerge}([l', r'], i)$ in Algorithm \ref{alg:main_iteration}.

The problem is now to find up to two ones in the bit vector
$\texttt{rightavailable}$ in the interval of indices 
$[\ell_{\$ \alpha}, r_{\$ \alpha}]$. To do this efficiently,
we maintain for each index $k$ in
$\texttt{rightavailable}$ the index
of the first one in $\texttt{rightavailable}[k+1..m]$, denoted
with $\texttt{next\_one}(k)$. 
If there are two ones in the interval 
$[\ell_{\$ \alpha}, r_{\$ \alpha}]$, then they can be found at
 $\texttt{next\_one}(\ell_{\$ \alpha}-1)$ and
$\texttt{next\_one}(\texttt{next\_one}(\ell_{\$ \alpha}-1))$.
The question now becomes, how do we maintain this information
efficiently?
In general, this is the problem of indexing a bit vector for
dynamic successor queries, for which there does not exist a constant
time solution using $O(n \log \sigma)$ space in the literature. 
However, in our case the vector
$\texttt{rightavailable}$ starts out filled with ones, and
once a one is changed to a zero, it will not change back for
the duration of the algorithm,
which allows us to have a simpler and 
more efficient data structure. 

Initially, $\texttt{next\_one}(k) = k+1$ for all $k < m$.
The last index does not have a successor, but it can
easily be handled as a special case. For clarity and brevity
we describe the rest of the process as if the special
case did not exist. When we update
$\texttt{rightavailable}(k) := \texttt0$, then we need to also
update $\texttt{next\_one}[k'] := \texttt{next\_one}(k)$
for all $k' < k$ such that $\texttt{rightavailable}[k'+1..k]$ contains
only zeros.
To do this efficiently, we store the value of
$\texttt{next\_one}$ only once for each sequence of
consecutive zeros in $\texttt{rightavailable}$, which allows
us to update the whole range at once. To keep track of
the sequences of consecutive zeros, we can use a union-find
data structure. A union-find data structure maintains
a partitioning of a set of elements into disjoint groups.
It supports the operations $\texttt{find}(x)$, which returns the
representative of the group containing $x$, 
and $\texttt{union}(x,y)$, which takes two representatives
and merges the groups containing them. 

We initialize the union-find structure such that
there is an element for every index 
in $\texttt{rightavailable}$, and we also initialize an
array $\texttt{next}[1..m]$ such that
$\texttt{next}[k] := k+1$ for all $k = 1, \ldots m$.
When a value at index $k$ is changed to a zero, we compute
$q := \texttt{next}[\texttt{find}(k)]$. Then we will
do $\texttt{union}(\texttt{find}(k), \texttt{find}(k-1))$
and if
$\texttt{rightavailable}[k+1] = 0$, we will do
$\texttt{union}(\texttt{find}(k), \texttt{find}(k+1))$.
Finally, we update
$\texttt{next}[\texttt{find}(k)] = q$. We can answer
queries for $\texttt{next\_one}(k)$ with 
$\texttt{next}[\texttt{find}(k)]$.



Whenever we find a pair of indices $i$ and $j$ such that 
$\texttt{leftavailable}[i] = 1$,
$\texttt{rightavailable}[j] = 1$ and 
$\texttt{leftend}[j] \neq i$, we add an edge from $s_i$ to $s_j$
by recording string $j$ as the successor
of string $i$ using arrays $\texttt{successor}[1..m]$ and
$\texttt{overlaplength}[1..m]$. We set $\texttt{successor}[j] = i$
and $\texttt{overlaplength}[j] = d_i$, where $d_i$ is the length
of the overlap of $s_i$ and $s_j$, and do the updates:
\begin{align*}
\texttt{leftavailable}[i] &:= 0 \\
\texttt{rightavailable}[j] &:= 0 \\
\texttt{leftend}[\texttt{rightend}[j]] &:= \texttt{leftend}[i] \\
\texttt{rightend}[\texttt{leftend}[i]] &:= \texttt{rightend}[j] \\
\end{align*}
Note that the arrays  $\texttt{leftend}$ and $\texttt{rightend}$ are only up to date for
the end points of the paths, but this is fine for the algorithm.
Finally we update the $\texttt{next}$ array with the union-find
structure using the process described earlier.
We stop iterating when we have done $m-1$ merges.
At the end, we have a Hamiltonian path over $G$, and we
form a superstring by merging the strings in the order
specified by the path.

\section{Time and Space Analysis} \label{sec:timespace}

The following space analysis is in terms of number of bits used.
We assume that the strings are binary encoded
such that each character takes $\lceil \log_2 \sigma \rceil$ bits.
A crucial observation is that we can afford to store a
constant number of $O(\log n)$ bit machine words for each distinct input string.

\begin{lemma} \label{lemma:mlogn}
Let there be $m$ \textbf{distinct} non-empty strings with combined length $n$ from 
an alphabet of size $\sigma > 1$. Then  $m \log n \in O(n \log \sigma)$.
\end{lemma}

\noindent Proof. Suppose $m \leq \sqrt n $. Then the Lemma is
clearly true, because:
$$m \log n \leq \sqrt n \log n \in O(n \log \sigma)$$ 
We now consider the remaining case $m \geq \sqrt n$, 
or equivalently $\log n \leq 2 \log m$. This means
$m \log n \leq 2 m \log m$, so it suffices to show
$m \log m \in O(n \log \sigma)$.

First, note that at least half of the strings have length at least
$\log (m) - 1$ bits. This is trivially true when $\log (m) -1 \leq 1$. When
$\log (m) - 1 \geq 2$, the number of distinct binary strings of length at most
$\log(m)-2$ bits is 
$$\sum_{i = 1}^{\lfloor \log(m)-2 \rfloor} 2^i
\leq 2^{\log (m) - 1} = \frac{1}{2} m
$$
Therefore indeed at least half of the strings have length of at least
$\log m - 1$ bits. The total length of the strings
is then at least $\frac{1}{2}m(\log m  - 1)$ bits.
Since the binary representation of all strings combined takes
$n \lceil \log_2 \sigma \rceil$ bits, we have 
$n \lceil \log_2 \sigma \rceil \geq \frac{1}{2}m(\log m  - 1)$, which implies
$m \log m \leq 2 n \lceil \log_2 \sigma \rceil + 1 \in O(n \log \sigma). \, \qed$

Next, we describe how to implement the suffix links and left extensions.
We will need to build the following data structures
for the concatenation of all input strings separated by
a separator character: 
\begin{itemize}
\item The Burrows-Wheeler transform, represented
as a wavelet tree with support for rank and select queries.
\item The $C$-array, which has length equal to the number of
characters in the concatenation, such that
$C[i]$ is the number of occurrences of characters with
lexicographic rank strictly less than $i$.
\item The balanced parenthesis representation of the suffix
tree topology with support for queries for leftmost leaf,
rightmost leaf and lowest common ancestor.
\end{itemize}
Note that in the concatenation of the strings, 
the alphabet size is increased by one
because of the added separator character, and the total length
of the data in characters is increased by $m$. However
this does not affect the asymptotic size of the data, because
$$(n+m) \log (\sigma + 1) \leq 2n (\log \sigma + 1)
\in \Theta (n \log \sigma)
$$
The three data structures can be built and represented
in $O(n \log \sigma)$ time and space \cite{belazzougui2014linear}.
Using these data structures we can implement the left extension
for lexicographic interval $[\ell, r]$ with the character $c$ by:
$$([\ell, r], c) \mapsto [C[c] 
+ \texttt{rank}_{BWT}(\ell,c), C[c] + \texttt{rank}_{BWT}(r,c)] $$
We can implement the suffix link for the right maximal 
string $c \alpha$ with the 
lexicographic interval $[\ell, r]$ by first computing
$$v = \texttt{lca}(\texttt{select}_{BWT} (c, \ell - C[c]),
\texttt{select}_{BWT} (c, r - C[c]))$$
and then
$$ 
[\ell, r] \mapsto [\texttt{leftmostleaf}(v), \texttt{rightmostleaf}(v)]
$$
This suffix link operation works as required for right-maximal strings by removing
the first character of the string, but
the behaviour on non-right-maximal strings is slightly different.
The lexicographic range of a non-right-maximal string is the same as the
lexicographic range of the shortest right-maximal string that has it as a prefix.
In other words, for a non-right-maximal string $c \alpha$, the operation maps the 
interval $[\ell_{c \alpha}, r_{c \alpha}]$ to the lexicographic interval
of the string $\alpha \beta$, where $\beta$ is the shortest right-extension
that makes $c \alpha \beta$ right-maximal. This behaviour allows
us to check the right-maximality of a substring $c \alpha$
given the lexicographic ranges $[\ell_\alpha, r_\alpha]$ and 
$[\ell_{c \alpha}, r_{c \alpha}]$ in the iterator initialization phase
of the algorithm as follows:

\begin{lemma}The substring
$c \alpha$ is right maximal if and only if the suffix link of 
$[\ell_{c \alpha}, r_{c \alpha}]$ is $[\ell_\alpha, r_\alpha]$.
\end{lemma}
\begin{proof}
As discussed above, the suffix link of $[\ell_{c \alpha}, r_{c \alpha}]$ 
maps to the lexicographic interval of the string 
$\alpha \beta$ where $\beta$ is the shortest right-extension
that makes $c \alpha \beta$ right-maximal.
Suppose first that $c \alpha$ is right-maximal. Then
$[\ell_{\alpha\beta}, r_{\alpha\beta}] = [\ell_\alpha, r_\alpha]$, because
$\beta$ is an empty string. Suppose on the contrary that $c \alpha$
is not right-maximal. Then 
$[\ell_{\alpha\beta}, r_{\alpha\beta}] \neq [\ell_\alpha, r_\alpha]$, because
$\alpha\beta$ and $\alpha$ are distinct right-maximal strings. $\square$
\end{proof}
Now we are ready to prove the time and space complexity of the whole algorithm.
\begin{theorem}
The algorithm in Section \ref{sec:algorithm} can
be implemented in $O(n \log \sigma)$ time
and $O(n \log \sigma)$ bits of space.
\end{theorem}

\begin{proof} The preprocessing to remove contained and duplicate strings
can be done in $O(n \log \sigma)$ time and space 
for example by building
an FM-index, and backward searching all input strings.

The algorithm executes
$O(n)$ left extensions and suffix links.
The time to take a suffix link is dominated by the time
do the select query, which is $O(\log \sigma)$, and 
the time to do a left extension is dominated by the time to
do a rank-query which is also $O(\log \sigma)$.
For each left extension the algorithm does, it has to access and modify
the union-find structure. Normally this would take amortized time
related to the inverse function of the Ackermann function
\cite{thomas2009introduction}, but in our case the amortized complexity
of the union-find operations can be made linear using the construction
of Gabow and Tarjan \cite{gabow1985linear}, because we know
that only elements corresponding to consecutive positions in the
array $\texttt{rightavailable}$ will be joined together. Therefore,
the time to do all left extensions, suffix links and updates to the
union-find data structure is $O(n \log \sigma)$.

Let us now turn to consider the space complexity. For each
input string, we have the quadruple $(i,\ell,r,d)$
of positive integers with value at most $n$. The quadruples
take space $3 m \log m + m \log n$. The
union-find structure of Gabow and Tarjan can be implemented in $O(m \log m)$ bits
of space \cite{gabow1985linear}. The bit vectors $\texttt{leftavailable}$
and $\texttt{rightavailable}$
take exactly 2$m$ bits, and the arrays $\texttt{successor}$, $\texttt{leftend}$,
$\texttt{rightend}$ and $\texttt{next}$ take $m \log m$ bits each. 
The array $\texttt{overlaplength}$ takes $m \log n$ bits of space. 
Summing up, in addition
to the data structures for the left extensions
and contractions, we have only $O(m \log n)$ bits of space, which
is $O(n \log \sigma)$ by Lemma \ref{lemma:mlogn}. $\qed$ \end{proof} 

\begin{figure}[htb!]
\makebox[\columnwidth]{\includegraphics[width=0.8\columnwidth]{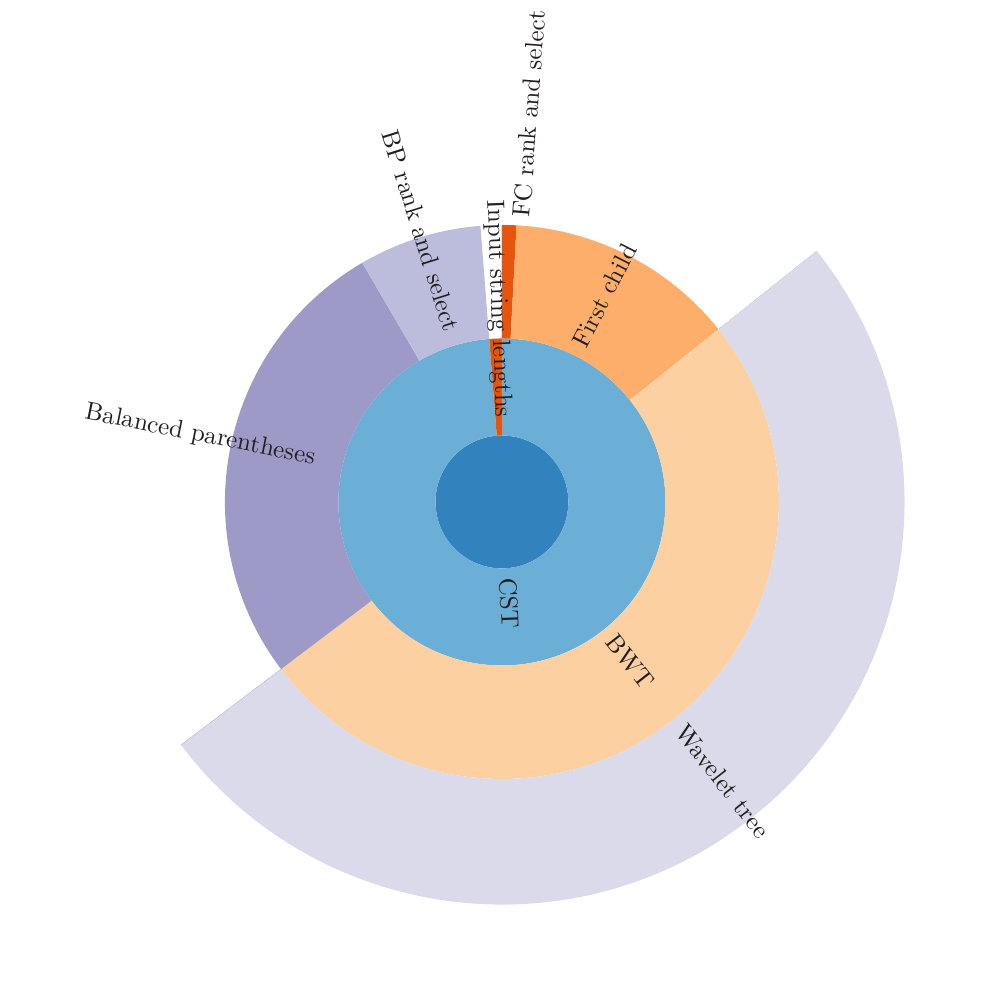}}
\caption{Memory breakdown of the data structures used by our implementation.
The plot was generated using the SDSL library.
Each sector angle represents the portion of the memory taken by the data structure
of the total memory of the inner data structure; areas have no special meaning.
Abbreviations: 
CST = compressed suffix tree, BWT = Burrows-Wheeler Transform,
BP = balanced parenthesis, FC = first child.
}
\label{fig:memory_pie}
\end{figure}

\begin{figure}[htb!]
	\sbox\figurebox{%
		\resizebox{\dimexpr.95\textwidth-1em}{!}{%
			\includegraphics[height = 3cm]{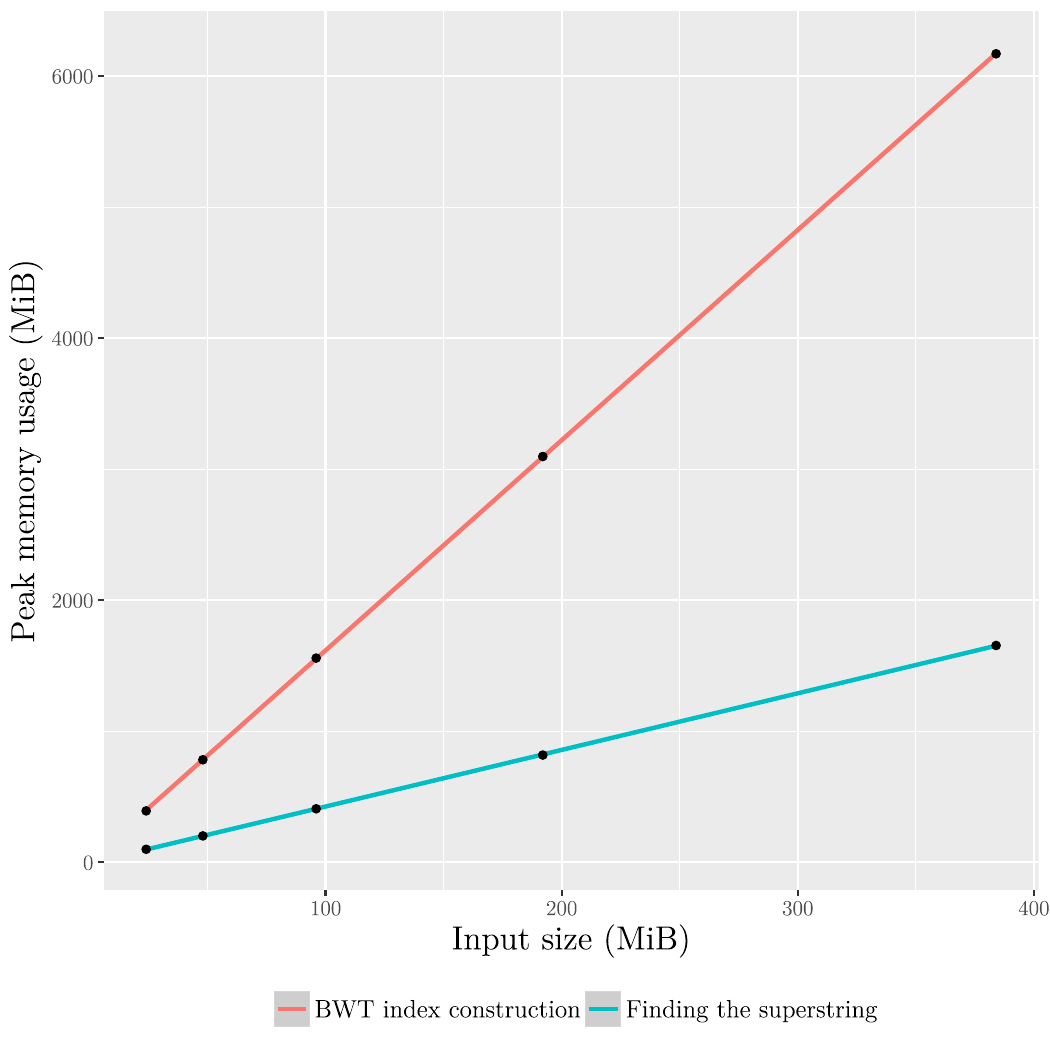}%
			\includegraphics[height = 3cm]{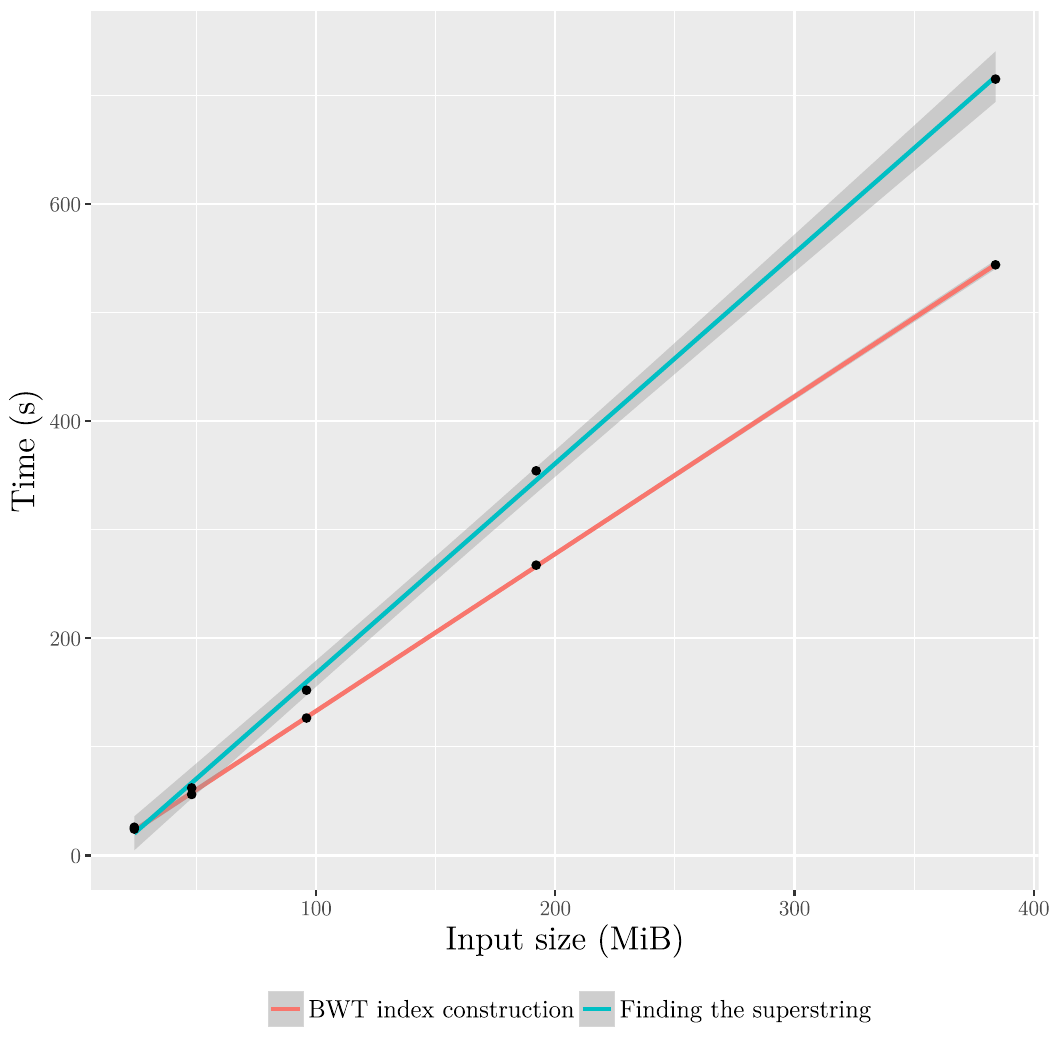}%
		}%
	}
	\setlength{\figureboxheight}{\ht\figurebox}

	\centering
	\subcaptionbox{Peak memory}{%
		\includegraphics[height = \figureboxheight]{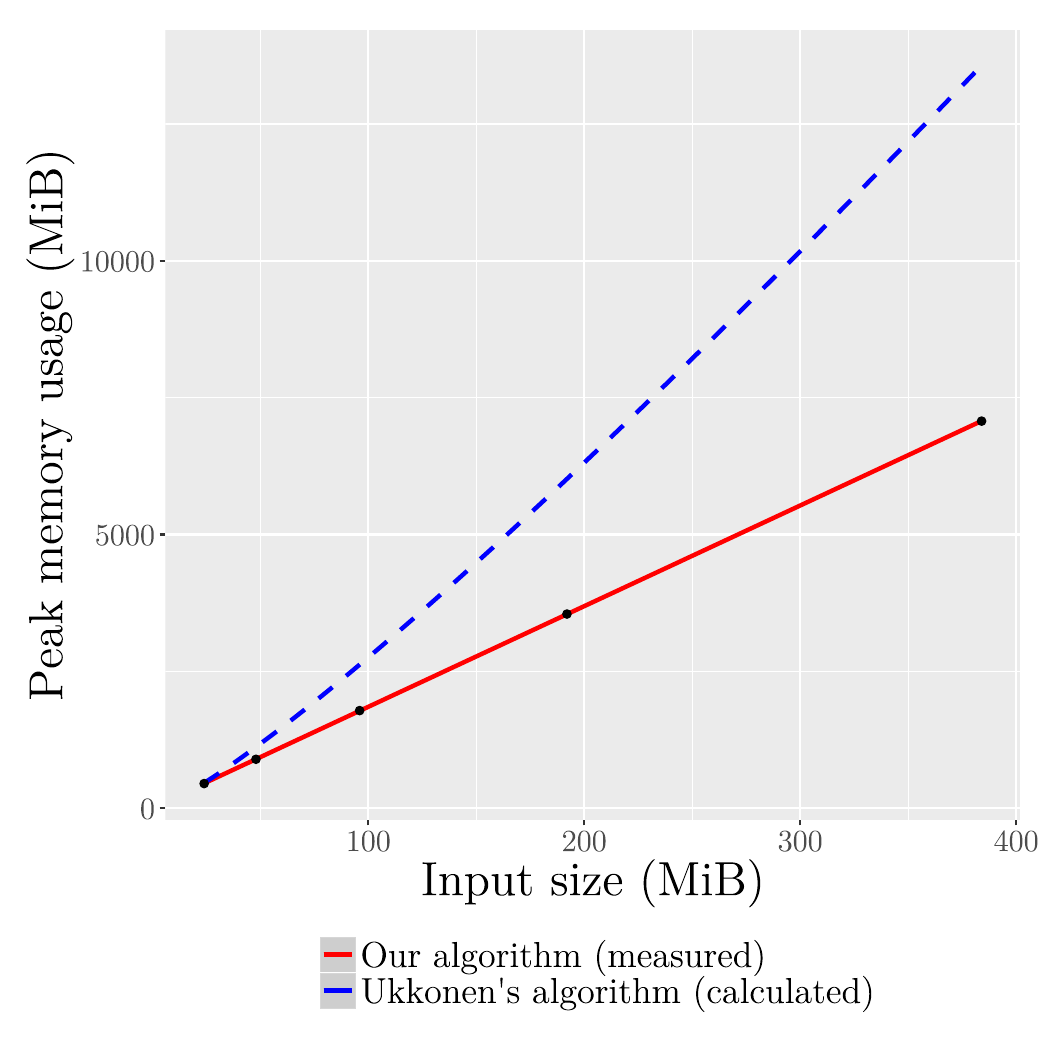}%
	}\quad
	\subcaptionbox{Time consumption}{
		\includegraphics[height = \figureboxheight]{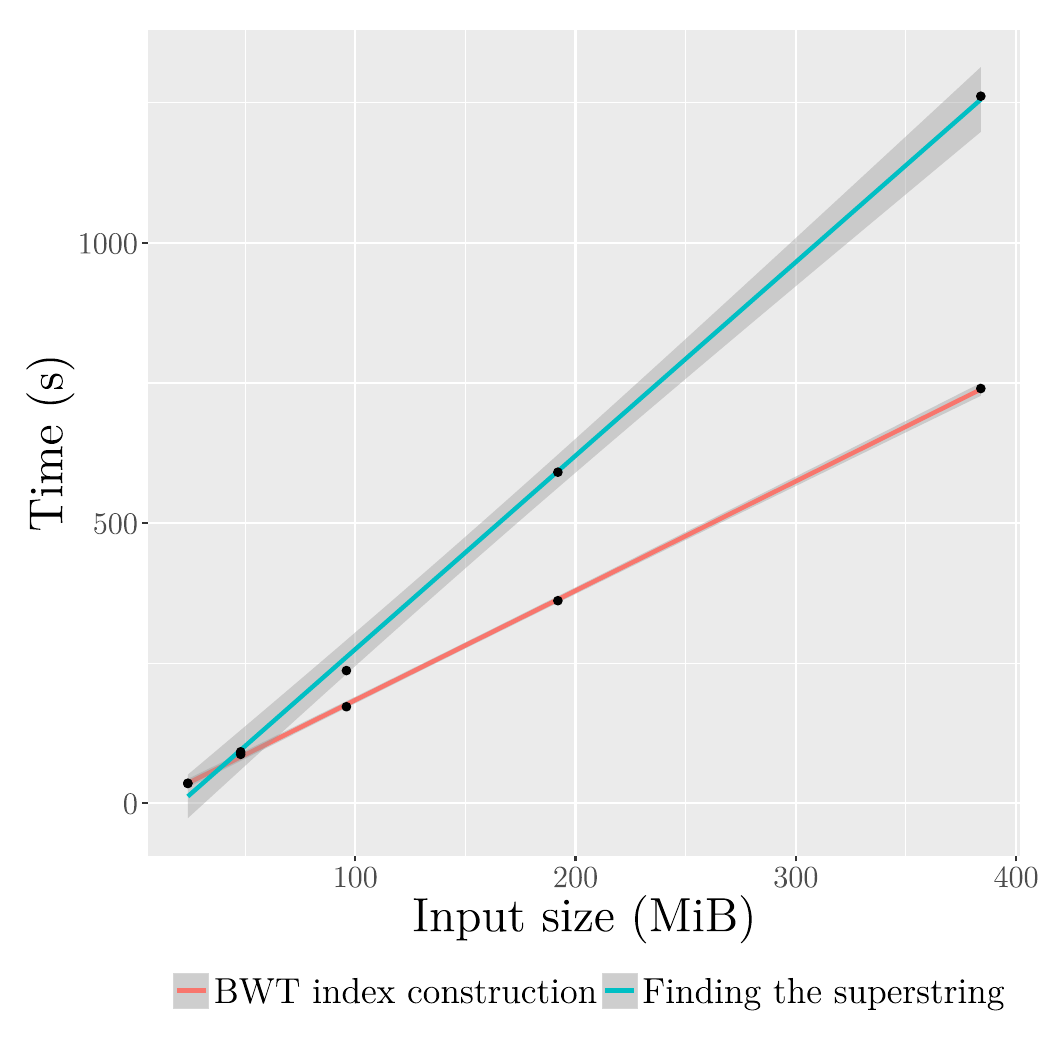}%
	}
	\caption{(a) The peak memory usage of our algorithm plotted against a conservative estimate of $4 n \log n$ bits of space needed by Ukkonen's Aho-Corasick based method. (b) the time usage of our algorithm for the two phases of the algorithm. The data points have been
	fitted with a least-squares linear model, and the grey band
	shows the 95\% confidence interval (large enough to be visible only for the second phase). The time and memory usage were measured using the /usr/bin/time command and the RSS value.}
	\label{fig:time_space}
\end{figure}

\begin{figure}[htb!]
	\sbox\figurebox{%
		\resizebox{\dimexpr.95\textwidth-1em}{!}{%
			\includegraphics[height = 3cm]{results/human_gut_sizes}%
			\includegraphics[height = 3cm]{results/human_gut_times}%
		}%
	}
	\setlength{\figureboxheight}{\ht\figurebox}

	\centering
	\subcaptionbox{Index construction}{%
		\includegraphics[height = \figureboxheight]{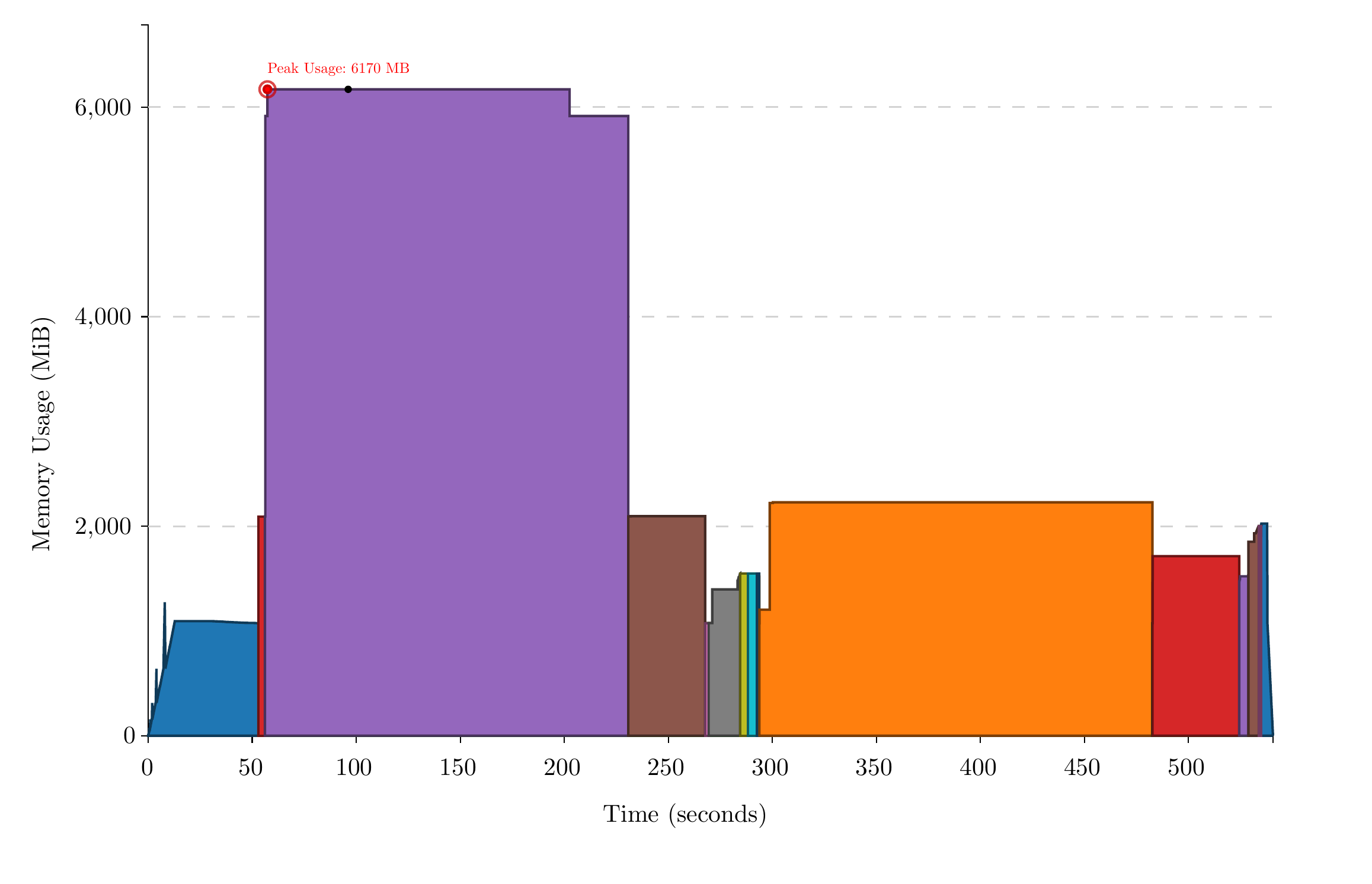}%
	}\quad
	\subcaptionbox{Superstring construction}{
		\includegraphics[height = \figureboxheight]{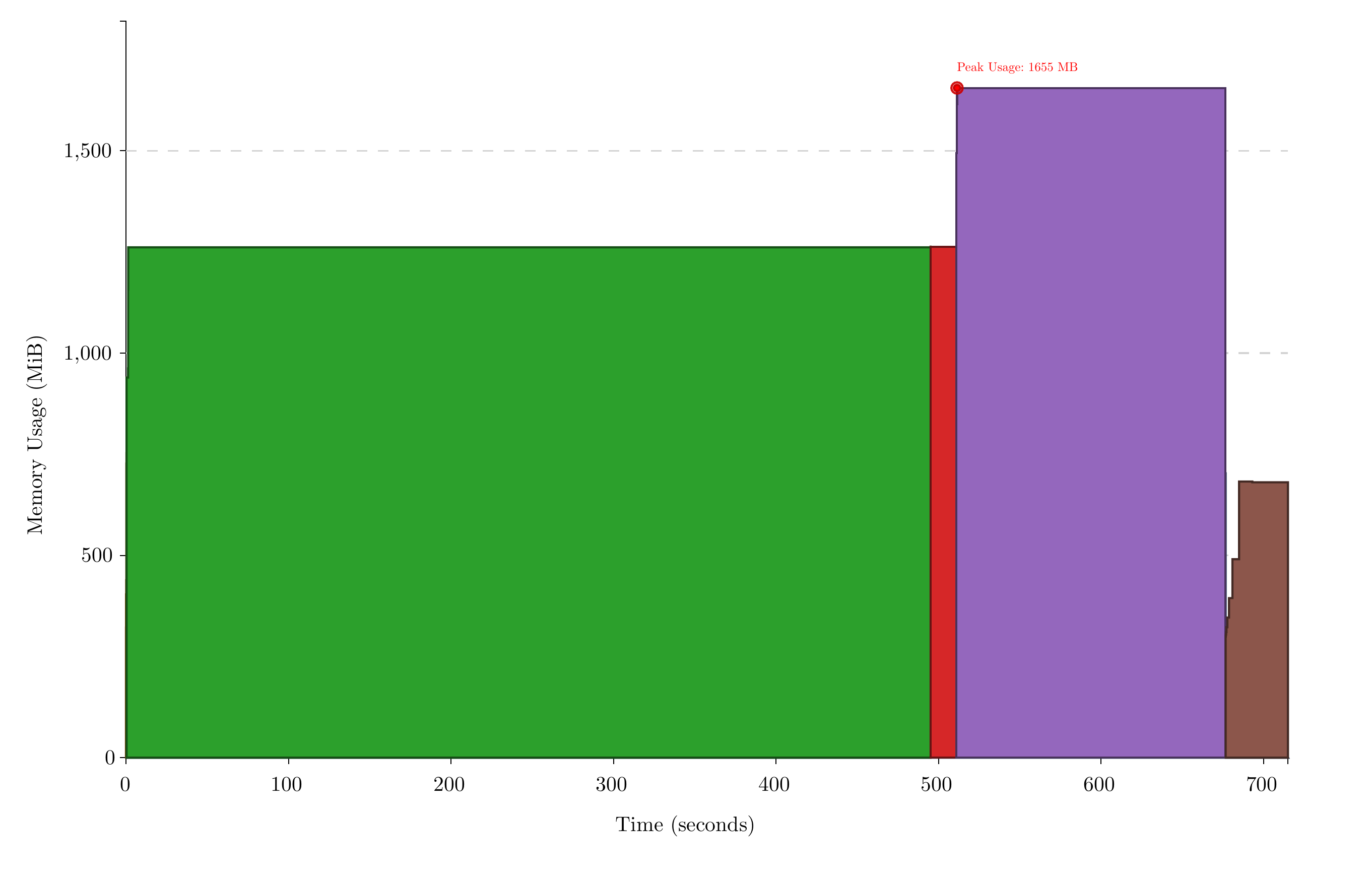}%
	}
	\caption{Subfigures (a) and (b) show the memory usage
	as a function of time for index construction and
	superstring construction, respectively. The peak in
	Figure (a) occurs during suffix array construction,
	and the peak in Figure (b) occurs during the iteration
	of prefix-suffix overlaps.}
	\label{fig:space_as_function_of_time}
\end{figure}

\section{Implementation}

The algorithm was implemented with the SDSL library \cite{gog2014theory}.
A compressed suffix tree that represents nodes as lexicographic intervals \cite{ohlebusch2010cst++} was 
used to implement the suffix links and left extensions.
Only the required parts of the 
suffix tree were built: the FM-index, balanced parentheses support and a bit vector that indicates the leftmost child node of each node. These data structures differ slightly from the description in Section
\ref{sec:timespace}, because they
were chosen for convenience as they were readily available 
in the SDSL library, and they should give very similar performance compared
to those used in the aforementioned Section. In particular,
the leftmost
child vector was needed to support suffix links, but we could manage
without it by using the operations on the balanced parenthesis support
described in Section \ref{sec:timespace}. Our implementation is
available at the URL 
$\texttt{https://github.com/tsnorri/compact-superstring}$

The input strings are first sorted with quicksort.
This introduces a $\log n$ factor to the time complexity,
but it is fast in practice.
The implementation then runs in two passes. 
First, exact duplicate strings are removed 
and the stripped compact suffix tree is built from the remaining strings. 
The main algorithm is 
implemented in the second part. The previously built stripped 
suffix tree is loaded into memory 
and is used to find the longest right-maximal suffix of each string
and to iterate the prefix-suffix overlaps.
Simultaneously, strings that are 
substrings of other strings are marked for exclusion from building the superstring. 


For testing, we took a metagenomic DNA sample from a human gut microbial
gene catalogue project \cite{qin2010human}, and sampled DNA fragments
to create five datasets with $2^{26 + i}$ characters respectively
for $i = 0,\ldots,4$. The alphabet of the sample was $\{A,C,G,T,N\}$.
Time and space usage for all generated datasets for both the index construction
phase and the superstring construction phase are plotted in Figure
\ref{fig:time_space}.
The machine used run Ubuntu Linux version 16.04.2 and
has 1.5 TB of RAM and four Intel Xeon CPU E7-4830 v3 processors (48 total cores, 2.10 GHz each). A breakdown
of the memory needed for the largest dataset for the different 
structures comprising the index is shown in Figure \ref{fig:memory_pie}. 

While we don't have an implementation of Ukkonen's greedy
superstring algorithm, have a conservative estimate
for how much space it would take. 
The algorithm needs at least the goto- and failure links for the 
Aho-Corasick automaton, which take
at least $2 n \log n$ bits total. The main algorithm 
uses linked lists named $L$ and $P$, which take at least 
$2 n \log n$ bits total. Therefore the space usage is at the very 
least $4 n \log n$. This estimate is plotted in
Figure \ref{fig:time_space}.

Figure \ref{fig:space_as_function_of_time} shows the space
usage of our algorithm in the largest test dataset 
as a function of time reported by the SDSL library. The peak memory usage of the whole algorithm occurs
during index construction, and more specifically during
the construction of a compressed suffix array. The 
SDSL library used this data structure
to build the BWT and the balanced parenthesis representation,
which makes the space usage unnecessarily high.
This could be improved by using more efficient algorithms to build 
the BWT and the balanced parenthesis representation of the suffix tree topology
\cite{belazzougui2014linear}. These could be plugged in to 
bring down the index construction memory. The peak
memory of the part of the algorithm which constructs the
superstring is only approximately 5 times the size
of the input in bits.

\section{Discussion}
We have shown a practical way to implement the greedy shortest
common superstring algorithm in $O(n \log \sigma)$ time and bits of space.
After index construction, the algorithm consists of two relatively independent parts: reporting
prefix-suffix overlaps in decreasing order of lengths, and maintaining
the overlap graph to prevent merging a string to one direction more than
once and the formation of cycles. The part which reports the
overlaps could also be done in other ways, such as using  compressed
suffix trees or arrays, or a succinct representation of the Aho-Corasick
automaton. The only difficult part is to avoid having to hold 
$\Omega(n)$ integers in memory at any given time. We believe it is
possible to engineer algorithms using these data structures to
achieve $O(n \log \sigma)$ space as well. 

Regrettably, we could not find any linear time implementations of Ukkonen's greedy
shortest common superstring algorithm for comparison.
There is an interesting implementation by Liu and Sýkora
\cite{liu2005sequential}, but it is too slow for our purposes because it involves computing all 
pairwise overlap lengths of the input strings to make better choices in resolving ties in the 
greedy choices. While their experiments indicate that this improves the quality of the
approximation, the time complexity is quadratic in the number of input strings.
Zaritsky and Sipper \cite{zaritsky2004preservation} also have an implementation of the
greedy algorithm, but it's not publicly available, and the focus of the paper is on
approximation quality, not performance.
As future work, it would be interesting to make a careful
implementation of Ukkonen's greedy algorithm, and compare
it to ours experimentally.



\section*{Acknowledgements}
We would like to thank anonymous reviewers for
improving the presentation of the paper.

\bibliographystyle{splncs03} 
\bibliography{references}

\end{document}